\documentclass{article}
\usepackage[utf8]{inputenc}
\usepackage[T1]{fontenc}

\usepackage{fullpage}
\usepackage{amsmath,amssymb,amsfonts,mathrsfs,amsthm}
\usepackage{graphicx}
\usepackage{pdfpages}

\usepackage{enumitem}

\usepackage{multirow}
\usepackage{varioref}
\usepackage{datetime}
\usepackage{mathtools}
\usepackage{ifthen}
\usepackage{stmaryrd}
\usepackage[h]{esvect}
\usepackage{array}
\usepackage{listings}
\usepackage{booktabs}
\usepackage[linesnumbered,lined,ruled,vlined]{algorithm2e}
\usepackage{makecell}
\usepackage{pdflscape}
\usepackage{subcaption}
\usepackage{thm-restate}
\usepackage[linkcolor=black,colorlinks=true,citecolor=black,filecolor=black]{hyperref}
\usepackage[capitalise]{cleveref}
\usepackage{stmaryrd}

\numberwithin{equation}{section}

\newtheorem{theorem}{Theorem}

\newtheorem{lemma}[theorem]{Lemma}

\newtheorem{conjecture}{Conjecture}

\theoremstyle{definition}
\newtheorem{definition}[theorem]{Definition}

\crefname{claim}{claim}{claims}
\crefname{observation}{observation}{observations}

\renewcommand{\O}{\ensuremath{\mathcal{O}}\xspace}

\newcommand{\R}{\ensuremath{\mathbb{R}}\xspace}

\renewcommand{\L}{\ensuremath{\mathcal{L}}\xspace}

\newcommand{\A}{\ensuremath{\mathcal{A}}\xspace}

\newcommand{\rotation}{\includegraphics[scale = 0.15]{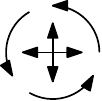}}
\newcommand{\translation}{\includegraphics[scale = 0.15, page =2]{figures/rotation.pdf}}

\newcommand{\NP}{\ensuremath{\textsf{NP}}\xspace}
\newcommand{\PSPACE}{\ensuremath{\textsf{PSPACE}}\xspace}
\newcommand{\ER}{\texorpdfstring{\ensuremath{\exists\R}}{ER}\xspace}

\newcommand{\var}[1]{\left\llbracket#1\right\rrbracket}

\newcommand{\Connector}{connector\xspace}
\newcommand{\unit}{\textsc{unit}\xspace}
\newcommand{\intersection}{\textsc{intersection}\xspace}
\newcommand{\stretchability}{\textsc{Stretchability}\xspace}
\newcommand{\UnitRec}{\textsc{Unit Segment Recognition}\xspace}
\newcommand{\PolyRec}{\textsc{$k$-PolyLine Recognition}\xspace}
\newcommand{\Rec}{\textsc{Recognition}\xspace}
\newcommand{\polyline}{polyline\xspace}
\newcommand{\polylines}{polylines\xspace}
\newcommand{\Polylines}{Polylines\xspace}
\DeclareMathOperator{\trace}{trace}

\title{Recognition of Unit Segment and Polyline Graphs is \ER-Complete}

\usepackage{authblk}
\author[1]{Michael Hoffmann}
\author[2]{Tillmann Miltzow}
\author[1]{Simon Weber}
\author[3]{Lasse Wulf}
\affil[1]{Department of Computer Science, ETH Z\"{u}rich, Switzerland}
\affil[2]{Department of Information and Computing Sciences, Utrecht University, The Netherlands}
\affil[3]{Institute of Discrete Mathematics, Graz University of Technology, Austria}
\date{}

\begin{document}

\maketitle

\begin{abstract}
      Given a set of objects $O$ in the plane, the corresponding intersection graph is defined as follows.
    Each object defines a vertex and an edge joins two vertices whenever the corresponding objects intersect.
    We study here the case of unit segments and \polylines with exactly $k$ bends.
    In the recognition problem, we are given a graph and want to decide whether the graph can be represented as an intersection graph of certain geometric objects.
    In previous work it was shown that various recognition problems are \ER-complete, leaving unit segments and \polylines among the few remaining natural cases where the recognition complexity remained open. 
    We show that recognition for both families of objects is \ER-complete.
\end{abstract}

\vfill

\begin{center}
    \includegraphics{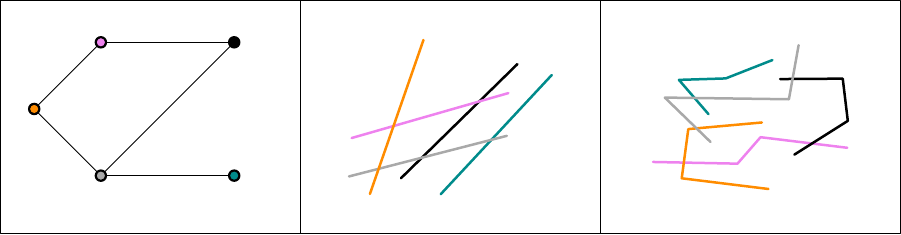}

    \bigskip
    
    A graph and two representations as an intersection graph, one using unit segments and one using \polylines.
\end{center}

\vfill

\paragraph{Acknowledgements.}

This research started at the \textit{19th Gremo Workshop on Open Problems} in Binn VS, Switzerland, in June 2022.
We thank the organizers for the invitation and for providing a very pleasant and inspiring working atmosphere.
Tillmann Miltzow is generously supported by the Netherlands Organisation for Scientific Research (NWO) under project  no. VI.Vidi.213.150. Lasse Wulf is generously supported by the Austrian Science Fund (FWF): W1230.

\newpage
\section{Introduction}
Many real-life problems can be mathematically described in the language of graphs.
For instance, consider all the cell towers in Switzerland.
We want to assign each tower a frequency such that no two towers that overlap in coverage use the same frequency.
This can be seen as a graph coloring problem:
Every cell tower becomes a vertex, overlap indicates an edge and a frequency assignment corresponds to a proper coloring of the vertices;
see \Cref{fig:switzerland}.

\begin{figure}[htb]
    \centering
    \includegraphics[scale=0.6]{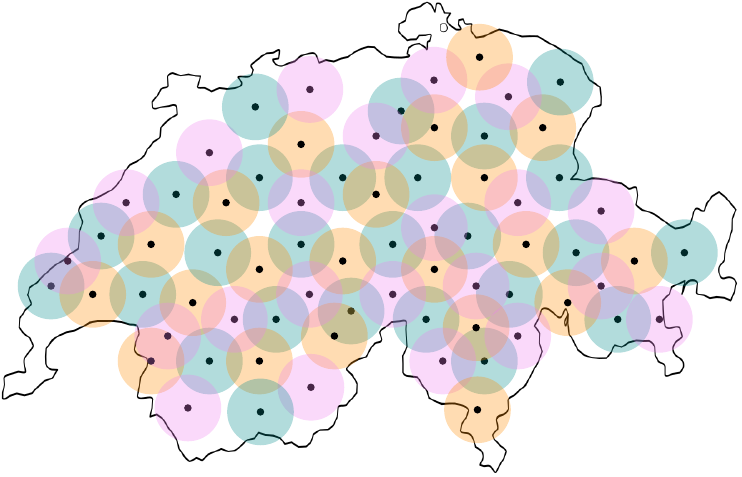}
    \caption{A fictional illustration of mobile coverage of Switzerland using cell towers.}
    \label{fig:switzerland}
\end{figure}

In many contexts, we have additional structure on the graph that may or may not help us to solve the underlying algorithmic problem.
For instance, it might be that the graph arises as an intersection graph of unit disks in the plane such as in the example above.
In that case, the coloring problem can be solved more efficiently~\cite{ColorDisks}, and there are better approximation algorithms for the clique problem~\cite{Clique-Disk}.
This motivates a systematic study of geometric intersection graphs.
Another motivation is mathematical curiosity. 
Simple geometric shapes are easily visualized and arguably very natural mathematical objects.
Studying the structural properties of intersection graphs gives insight into those geometric shapes and their possible intersection patterns.

It is known for a host of geometric shapes that it is \ER-complete to recognize their intersection graphs~\cite{Cardinal2018_Intersection,kratochvil1994intersection,matousek2014intersection,mcdiarmid2013integer}.
The class \ER consists of all of those problems that are polynomial-time equivalent to deciding whether a given system of polynomial inequalities and equations with integer coefficients has a real solution. We will introduce \ER in more detail below.

In this work, we focus on two geometric objects; unit segments and \polylines with exactly $k$ bends.
Although we consider both types of geometric objects natural and well studied, 
the complexity of their recognition problem was left open.
We close this gap by showing that both recognition problems are \ER-complete.


\subsection{Definition and Results}

We define geometric intersection graphs and the corresponding recognition problem.

\paragraph{Intersection graphs.}
Given a finite set of geometric objects $O$, we denote by $G(O)= (V,E)$, the corresponding \textit{intersection graph}.
The set of vertices is the set of objects ($V = O$) and  two objects $u$ and $v$ are adjacent ($uv\in E$) if 
they intersect ($u\cap v \neq \emptyset$). 
We are interested in intersection graphs that come 
from different families of geometric objects.

\paragraph{Families of geometric objects.} 
Examples for a family of geometric objects are  segments, disks, unit disks, unit segments, rays, and convex sets, to name a few of the most common ones.
In general, 
 given a geometric body $O \subseteq \R^2$ we denote by $O^{\translation}$  the family of all translates of $O$.
Similarly, we denote by $O^{\rotation }$  the family of all translates and rotations of $O$.
For example, the family of all unit segments can be denoted as $u^{\rotation}$, where $u$ is a unit segment.

\paragraph{Graph classes.} 
Classes of geometric objects $\O$ naturally give rise to classes of graphs $C(\O)$:
Given a family of geometric objects \O, we denote by $C(\O)$ the class of graphs that can be formed by taking the intersection graph of a finite subset from $\O$.

\paragraph{Recognition.}
If we are given a graph, we can ask if this graph belongs to a geometric graph class. 
Formally, let $C$ be a fixed graph class, then the recognition problem for $C$ is defined as follows.
As input, we receive a graph $G$ and we have to decide whether $G\in C$.
We denote the corresponding algorithmic problem by $\Rec(C)$.
For example, the problem of recognizing unit segment graphs is denoted by $\Rec(C(u^{\rotation}))$.
We will use the term \UnitRec for this problem.
Furthermore, we define \PolyRec as the recognition problem of intersection graphs of \polylines with $k$ bends.

\paragraph{Realizations.}
We can also say that $\Rec(C(\O))$ asks about the existence of a 
\textit{representation} of a given graph. 
A representation or \textit{realization} of a graph $G$ using a family of objects $\O$ 
is a function $r:V\mapsto \O$ such that $r(v)\cap r(w)\not=\emptyset \Longleftrightarrow vw\in E$.
For simplicity, for a set $V'\subseteq V$, we define $r(V')=\bigcup_{v\in V'} r(v)$.

\paragraph{Results.}
We show \ER-completeness of the recognition problems of two very natural geometric graph classes.

\begin{theorem}
    \label{thm:unit}
    \UnitRec is \ER-complete.
\end{theorem}

\begin{theorem}
    \label{thm:polylines}
    \PolyRec is \ER-complete,
    for any fixed $k\geq 0$.
\end{theorem}

\subsection{Discussion}
In this section, we discuss strengths and limitations of our results from different perspectives. To supply the appropriate context, we give a comprehensive list of important geometric graph classes and the current knowledge about the complexity of their recognition problems in \Cref{tab:Intersection}.

\begin{table}[ht]
   \caption{Classes of Geometric Intersection Graphs and their algorithmic complexity.} 
   \label{tab:Intersection}
   \begin{tabular}{|l|c|c|}
   \hline
   \multicolumn{1}{|c|}{Intersection graphs of} & \multicolumn{1}{c|}{Complexity} & \multicolumn{1}{c|}{Source} \\ \hline
   circle chords & polynomial & Spinrad~\cite{S94} \\  
   (unit) interval & polynomial & Booth and Lueker~\cite{BL76} \\
   string & \NP-complete & Kratochv{\'\i}l~\cite{StringHard}, Schaefer et al.\  \cite{schaefer2003recognizing} \\  
   outerstring & \NP-complete & Kratochv{\'\i}l~\cite{K91} (see also Rok and Walczak~\cite{RW19}) \\
   $C^{\translation}$, $C$ convex polygon & \NP-complete & 
   M\"{u}ller et al.~\cite{MvLvL13}, Kratochv{\'\i}l, Matou{\v{s}}ek~\cite{kratochvil1989np}\\
   (unit) disks & \ER-complete & McDiarmid and M{\"u}ller \cite{mcdiarmid2013integer}\\  
   convex sets & \ER-complete & Schaefer~\cite{Schaefer2010} \\  
   (downwards) rays & \ER-complete & Cardinal et al.~\cite{Cardinal2018_Intersection} \\  
   outer segments & \ER-complete & Cardinal et al.~\cite{Cardinal2018_Intersection} \\  
   segments & \ER-complete & Kratochv{\'\i}l and Matou{\v{s}}ek~\cite{kratochvil1994intersection} (see also~\cite{matousek2014intersection,Schaefer2010})\\  
   unit balls & \ER-complete & Kang and M{\"u}ller~\cite{Kang2012_Sphere} \\ 
   segments in 3D & \ER-complete & Evans et al.\ \cite{DBLP:conf/gd/EvansRSS019}\\
   \hline
   \end{tabular}
\end{table}

\begin{figure}[t!]
    \centering
    \includegraphics[scale=1]{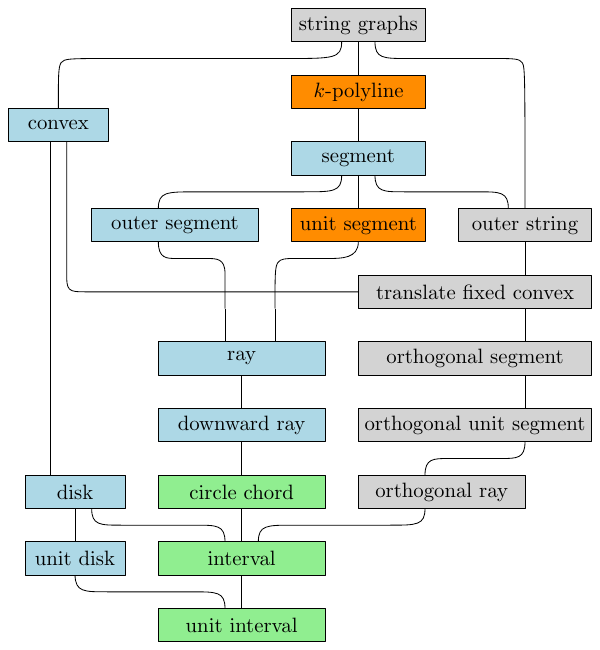}
    \caption{
    Each box represents a different geometric intersection graph class. Upward edges between classes express containment.
    Those marked in green can be recognized in polynomial time. 
    Those in blue are known to be \ER-complete. 
    The ones in gray are \NP-complete, 
    and the orange ones are the new results presented in this paper.
}
    \label{fig:Graph-Classes}
\end{figure}

\paragraph{Refining the hierarchy.} We see our main contribution in refining the hierarchy of geometric graph classes for which recognition complexity is known. 
    Both unit segments as well as \polylines with $k$ bends are natural objects that are well studied in the literature. However, the recognition of the corresponding graph classes has not been studied previously.  \Polylines with an unbounded number of bends are equivalent to strings (it is possible to show that polylines with an unbounded number of bends are as versatile as strings with respect to the types of graphs that they can represent, since the number of intersections of any two strings can always be bounded from above~\cite{schaefer2003recognizing,schaefer2004decidability}), while \polylines with $0$ bends are just segments. \Polylines with $k$ bends thus naturally slot in between strings and segments, and their corresponding graph class is thus also an intermediate class between the class of segment intersection graphs and string graphs, as can be seen in \Cref{fig:Graph-Classes}. By showing that recognition for \polylines with $k$ bends is \ER-complete for all constant $k$, we see that the switch from \ER-completeness (segment intersection graphs) to \NP-membership (string graphs) really only happens once $k$ is infinite. Similarly, unit segment intersection graphs slot in between segment and ray intersection graphs. 
    Intuitively, it is natural to expect that recognition of a class intermediate to two classes that are \ER-hard to recognize is \ER-hard again (although in general this intuition is wrong). Our \Cref{thm:unit} confirms this intuition in the case of $k$-bend polylines with $k \in \{1,2,\dots\}\cup\{\infty\}$.

\paragraph{Large coordinates.}
    One of the consequences of \ER-completeness is that
    there are no known ways to encode a solution (a representation) in polynomially many bits.
    Some representable graphs may only be representable by objects with irrational coordinates, or by rational coordinates with nominator and denominator of size 
    at least $2^{2^{n^c}}$, for some fixed $c > 0$. 
    In other words, the numbers to describe the position might need 
    to be doubly exponentially large~\cite{matousek2014intersection} for some graphs.
    For ``flexible'' objects such as polylines, rational solutions can always be obtained by slightly perturbing the representation. For more ``sturdy'' objects such as unit segments this may not be possible, however it is known that for example unit disks admit rational solutions despite their inflexibility~\cite{mcdiarmid2013integer}.

\paragraph{Unraveling the broader story.}
    Given the picture of \Cref{fig:Graph-Classes}, we aim to get a better understanding of 
    when geometric recognition problems are \ER-complete and when they are contained in \NP.
    \Cref{fig:Graph-Classes} indicates that \ER-hardness comes from objects that are complicated enough to avoid a complete combinatorial characterization. 
    Unit interval graphs, interval graphs and circle chord graphs admit such a characterization that in turn can be used to develop a polynomial-time algorithm for their recognition.
    On the other hand, if the geometric objects are too broad, the recognition problem is in \NP.
    The prime example is string graphs. 
    For string graphs, it is sufficient to know the planar graph given by the intersection pattern of the strings, which is purely combinatorial information that does not care anymore about precise coordinates.
    Despite the fact that there are graphs that need an exponential number of intersections, it is possible to find a polynomial-size witness~\cite{schaefer2003recognizing} and thus we do not have \ER-hardness.
    We want to summarize this as: recognition problems are \ER-complete if the underlying family of geometric objects is at a sweet spot of neither being too simplistic nor too flexible.

    When we consider \Cref{tab:Intersection} we observe 
    two different types of \ER-complete families.
    The first type of family
    encapsulates all rotations $O^{\rotation}$ of a given object $O$ (i.e., segments, rays, unit segments etc.).
    The second type of family contains translates and possibly homothets of geometric objects that have some curvature themselves (i.e., disks and unit disks). 
    However in case we fix a specific object without curvature, i.e., a polygon, and consider all translations of it, then the recognition problem lies in \NP~\cite{MvLvL13}.
    Therefore, broadly speaking, curvature or rotation seem to be properties needed for \ER-completeness and the lack of it seems to imply \NP{}-membership. We wish to capture one part of this intuition in the following conjecture:

    \begin{conjecture}
    \label{con:rotation}
        Let $O$ be a convex body in the plane with at least two distinct points.
        Then $\Rec(O^{\rotation})$ is \ER-complete.
    \end{conjecture}

    We wonder if our intuition on the curvature of geometric objects could be generalized as well. It seems plausible that $\Rec(O^{\translation})$ is \ER-complete if and only if $O$ has curvature.

\paragraph{Restricted geometric graph classes.}  
    
A classic way to make an algorithmic problem easier is to consider it only on a restricted input.
In the context of \ER and geometric graph recognition, we would like to mention the work by Schaefer~\cite{schaefer2021complexity}, who showed that 
the recognition problem of segments is already \ER-complete for graphs of bounded degree.
We conjecture that the same is true for unit segments and \polylines. However, this does not follow from our reductions since we create graphs of unbounded degree.

\begin{conjecture}
    \label{con:BoundedDegree}
        \UnitRec is \ER-complete already for graphs of bounded degree.
\end{conjecture}

Once this would be established it would be interesting to find the exact degree threshold when the problem switches from being difficult to being easy.

\subsection{Existential Theory of the Reals.}
The class of the existential theory of the reals \ER (pronounced as `ER') is a complexity class defined through its canonical problem ETR, which also stands for Existential Theory of the Reals. 
In this problem we are given a  sentence of the form \mbox{$\exists x_1, \ldots, x_n \in \R : \Phi(x_1, \ldots, x_n)$},
where~$\Phi$ is a well-formed quantifier-free formula consisting of the symbols $\{0, 1, x_1, \ldots, x_n, +, \cdot, \geq, >, \wedge, \vee, \neg\}$, and the goal is to check whether this sentence is true.

The class \ER is the family of all problems that admit a polynomial-time many-one reduction to ETR.
It is known that $\NP \subseteq \ER \subseteq \PSPACE$~\cite{canny1988some}. 
Note that there are oracle separations that indicate that $\ER$ and $\PSPACE$ are distinct~\cite{OracleSeparation}.
The reason that \ER is an important complexity class is that a number of common problems, mainly in computational geometry, have been shown to be complete for this class.
The compendium by Schaefer, Cardinal and Miltzow~\cite{ERcompendium} give a comprehensive overview.
Schaefer established the current name and pointed out first that
several known \NP-hardness reductions actually imply \ER-completeness~\cite{Schaefer2010}.
Early examples are related to recognition of geometric structures:
points in the plane~\cite{mnev1988universality,Shor1991_Stretchability},
geometric linkages~\cite{abel,schaefer2013realizability},
segment graphs~\cite{kratochvil1994intersection,matousek2014intersection},
unit disk graphs~\cite{Kang2012_Sphere,mcdiarmid2013integer},
ray intersection graphs~\cite{Cardinal2018_Intersection}, and
point visibility graphs~\cite{Cardinal2018_Intersection}.
In general, the complexity class is more established in the graph drawing community~\cite{Dobbins2018_AreaUniversality,erickson2019optimal,AnnaPreparation,schaefer2021complexity}.
Yet, it is also relevant for studying polytopes~\cite{NestedPolytopesER,richter1995realization}, Nash equilibria~\cite{berthelsen2019computational,bilo2016catalog,bilo2017existential,garg2015etr,Schaefer-ETR}, and matrix factorization problems~\cite{chistikov_et_al:LIPIcs:2016:6238,TensorRank,shitov2016universality,Shitov16a}.
Other \ER-complete problems are the Art Gallery problem~\cite{ARTETR,Stade22}, covering polygons with convex polygons~\cite{abrahamsen2021covering}, geometric packing~\cite{etrPacking} and 
training neural networks~\cite{abrahamsen2021training,TrainFullyNeurons}.

\subsection{Proof Techniques}
\label{sec:prelim}

\begin{figure}[tb]
    \centering
    \includegraphics[]{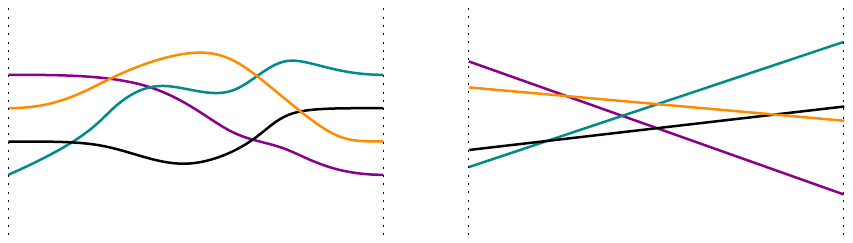}
    \caption{The pseudoline arrangement on the left is combinatorially equivalent to the (truncated) line arrangement on the right; hence, it is stretchable.}
    \label{fig:pseudolines}
\end{figure}

    The techniques used in this paper are similar
    to previous work. \ER-membership can be established straightforwardly by constructing concrete formulae or invoking a characterization of \ER using \emph{real} verification algorithms~\cite{SmoothingGap}, similar to the characterization of \NP. This proof can be found in \Cref{sec:membership}.
    
    For \ER-hardness, we are in essence reducing from the \stretchability problem. In this problem, we are given a simple pseudoline arrangement as an input, and the question is whether this arrangement is stretchable. 
    A pseudoline arrangement is a set of curves that are $x$-monotone.
    Furthermore, any two curves intersect exactly once and
    no three curves meet in a single point.
    We assume that there exist two vertical lines on which each curve starts and ends. The problem is to determine whether there exists a combinatorially equivalent (truncated) line arrangement. See \Cref{fig:pseudolines} for an example.

    For the reduction establishing \ER-hardness of \UnitRec, we are given a pseudoline arrangement \A, and we construct a graph that is representable by unit segments (or \polylines with $k$ bends) if and only if \A is stretchable. This graph is created by enhancing \A with more curves and taking their intersection graph. Due to the way these additional curves are constructed, they can always be drawn as unit segments in any line arrangement $\A'$ combinatorially equivalent to \A, if $\A'$ is first ``squeezed'' into a canonical form.

    For the other direction of the equivalence we need to show that if the graph is representable by unit segments, then \A is stretchable. To aid with this proof, we establish helper lemmas for both unit segments and \polylines stating that cycles in a graph enforce a certain order of intersections of objects around a closed curve in any realization of that graph. The details of these helper lemmas can be found in \Cref{sec:Cycle}.

    Using these additional lemmas, we then show that the unit segments corresponding to the vertices originally representing our pseudolines must have the same combinatorial structure as our original arrangement \A. To ensure this, we add so-called \emph{probes} during the construction of our graph $G$. The ideas of order-enforcing cycles and probes have already been used in a different context~\cite{Cardinal2018_Intersection}.

    The main idea to show hardness of \PolyRec is to enhance the construction for \UnitRec by an additional region where the \polylines representing our pseudolines have to make their $k$ bends. This then ensures that the polylines have no bends in the crucial part representing the pseudoline arrangement \A, and the same arguments as for unit segments can be used to show that this witnesses the stretchability of the original pseudoline arrangement.


\section{Cycle Representations}
\label{sec:Cycle}
In our reduction, we will construct a graph $G$  that contains a cycle.
The cycle helps us to enforce a certain structure on any geometric representation of $G$.
The arguments in this section merely use that our geometric objects are Jordan arcs who are piecewise differentiable, and are not specific to either unit segments or \polylines.

We first introduce some notion about realizations of induced cycles by unit segments, and then explain how to generalize these notions to all Jordan arcs.
Let $G$ be a graph, and let $C\subseteq V(G)$ be a set of vertices such that the induced graph $G[C]$ is a cycle $(c_1,\ldots,c_n)$ of $n\geq 4$ vertices. Consider now a given geometric representation of $G$ as an intersection graph of unit segments. We observe that the set $r(C)$ splits the plane into exactly two regions, since it is a set of unit segments whose intersection pattern is specified by the cycle $C$.
We can then define two Jordan curves ``tracing'' the representation $r(C)$ sufficiently close, one on  the inside and one on the outside of $r(C)$.
By sufficiently close we mean that the curves are close enough to $r(C)$ such that (i) between the two curves no other object $o\not\in r(C)$ starts or ends, (ii) no proper crossing of two such other objects occurs, and (iii) every crossing of another object $o$ with one of the curves implies that $o$ also intersects $r(C)$ within a small $\epsilon$-ball centered at the intersection with the curve.
We call the two curves \textit{interior boundary curve} and \textit{exterior boundary curve}.
See \Cref{fig:Boundary-Curves} (left column) for an illustration. Given either the interior or exterior boundary curve, we can record the \emph{order of traced elements}, i.e., the order of elements $c\in C$ for which the boundary curve is close to $r(c)$. We call this order the \emph{trace} of the curve.

\begin{figure}[tb]
    \centering
    \includegraphics[page = 2]{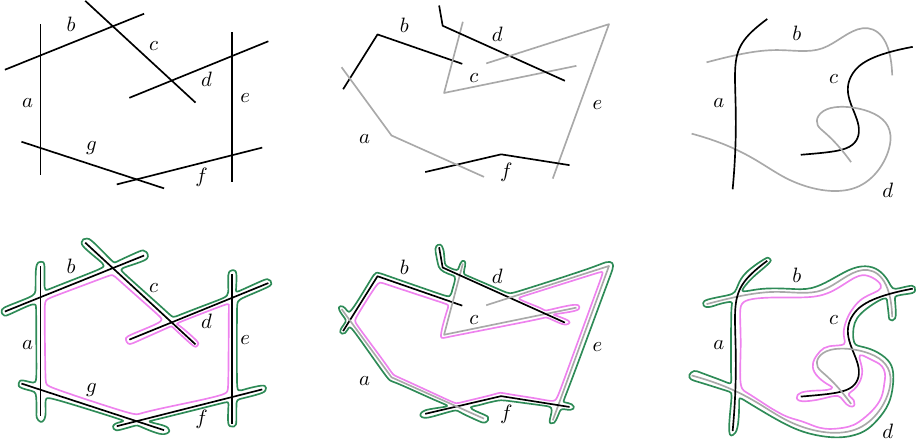}
    \caption{
    Top row: Representations $r(C)$ by unit segments, $1$-polylines and strings, respectively.
    Bottom row: The representation of the top row with two boundary curves added; one on the interior and one on the exterior of the representation.
    Left: The interior curve traces the Jordan arcs $abcdcdefg$. The exterior curve traces $ababcbcdedefefgfgag$.
    Middle:
    The interior curve traces the Jordan arcs $abcdcdef$. The exterior curve traces $ababcdcdefefaf$.
    Right: The interior curve traces the Jordan arcs $abcdcdcd$. The exterior curve traces $ababcbcdad$.
    }
    \label{fig:Boundary-Curves}
\end{figure}

In the more general case that the geometric realization of $G$ is given in terms of Jordan curves instead of only unit segments, we can sometimes extend the same notion of interior and exterior bounding curve and their trace, as showcased in \Cref{fig:Boundary-Curves} (middle and right column).
However, sometimes the situation is more complicated, as showcased in the following example: 
Consider \Cref{fig:counterexample-inner-outer-boundary}, which depicts a cycle $C$ on 6 vertices, and a possible realization $r(C)$ with Jordan arcs. Two cells $F,F'$ of the arrangement are highlighted. 
Now, note that the outer cell of the arrangement only touches the Jordan curves corresponding to vertices 1,5,6. 
This might be considered counter-intuitive, since one could maybe expect the outer boundary curve to touch all Jordan curves for vertices 1  \dots 6.
Furthermore, it is not clear what should be the interior bounding curve, since there are two viable candidates, either inside cell $F$ or cell $F'$. 
Since it is not clear what should be the interior or exterior boundary curve in this case, let us simply make the following definition: For any cell $F$ in the arrangement, 
we define the \emph{boundary curve of $F$} to be the curve that closely follows the boundary of $F$ in the arrangement $r(C)$, such that the same properties (i) -- (iii) from above are met. 
Note that since we assumed our Jordan curves to be piecewise differentiable, it is possible to follow the boundary of $F$ close enough such that properties (i) -- (iii) are met.
Analogously to above, we define the \emph{trace of $F$}, denoted by $\trace(F)$ as the order of encountered elements that are close to the boundary curve of $F$.

The main insight we offer in this section is that even in the more general case, the trace of $F$ has a very restricted structure. 
We show in the following lemma that if the cycle $C$ has $n$ vertices, then for any arbitrary cell $F$ of size at least 5 it is always possible to partition the sequence $\trace(F)$ into $n$ pieces, such that the $i$-th piece is either empty or only contains the symbols $i, i+1$ (mod n).
In particular, this implies that the order of elements along the trace of $F$ in any geometric realization is fixed, up to possible displacement of either +0/+1 (and up to possible cyclic shift or reflection). Formally, we have the following statement. For notational convenience, all cycle symbols in the lemma are to be understood modulo $n$ and start from 1, that is $n+1 = 1$ and $1 - 1 = n$.

\begin{figure}
    \centering
    \includegraphics[width=0.5\linewidth]{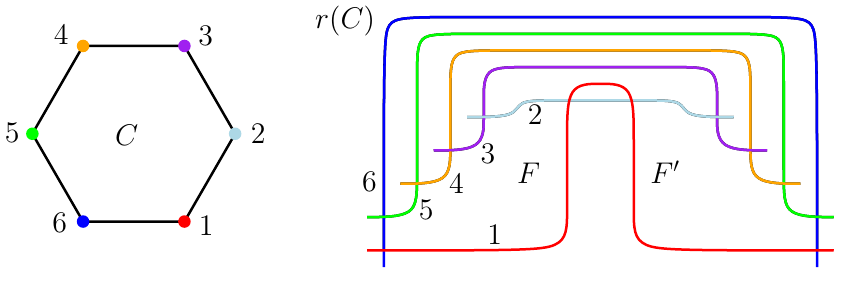}
    \caption{A geometric realization of the cycle on 6 vertices using Jordan arcs.}
    \label{fig:counterexample-inner-outer-boundary}
\end{figure}

\begin{lemma}
\label{lem:structureboundarycurve}
    Let $C \subseteq V(G)$ be an induced cycle of length $n \geq 4$ in the graph $G$, and let $r$ be a geometric realization of $G$ with Jordan curves. Let $F$ be any cell in the arrangement $r(C)$, and let $s := \trace(F)$ be its trace. If the sequence $s$ contains at least 5 different symbols, then there is a cyclic shift and possible reflection $\tilde s$ of $s$, such that we can rewrite $\tilde s = s_1s_2\dots s_n$ such that each $s_i \in \{i,i+1\}^\star$ is nonempty and contains only symbols $i,i+1$.
\end{lemma}
\begin{proof}
    We fist start by proving a weaker statement. We claim that there exist $n$ sequences $\overline s_1,\dots,\overline s_n$ such that $\overline s_i \in \{i-1,i,i+1\}$ and such that $\overline s_1,\dots,\overline s_n$ together "cover" the cyclic sequence $s$. 
    More precisely, let $k := |s|$ be the length of the trace, and consider a cycle $D$ on $k$ vertices (not related to the graph $G$). We can consider the cycle $D$ as labeled with the cyclic sequence $s = \trace(F)$.
    Our goal is to associate to each $i$ a cyclic arc $A_i$ of the cycle $D$, such that $|A_i| = \overline s_i$, and such that if we label the arc $A_i$ with the sequence $\overline s_i$ for each $i$, then in all the places where the arcs overlap, the strings agree, and so that all arcs together completely cover $D$ and their labels together result in the sequence $s$. Formally, we make the following claim: It is possible to find $\overline s_1,\dots,\overline s_n$ such that
    \begin{itemize}
        \item $\overline s_i \in \{i-1,i,i+1\}^\star$ for all $i=1,\dots,n$.
        \item $\overline s_1,\dots, \overline s_n$ together cover $s$.
        \item If $s = \trace(F)$ does not contain any symbol $i$, then $\overline s_i = \varepsilon$ is the empty sequence
        \item If $s = \trace(F)$ contains the symbol $i$, then $\overline s_i$ contains all occurrences of symbol $i$ and furthermore has symbol $i$ as its first and last letter.
    \end{itemize}
    We now prove the claim: Consider for each $i$ all occurrences of symbol $i$ on the cycle. 
    The situation is depicted in \Cref{fig:lem-boundary-curve-1} for $i=1$. 
    The set of occurrences of symbol $i$ splits the cycle into segments. 
    Consider all occurrences of symbols $X_i := \{1,\dots,n\} \setminus \{i-1,i,i+1\}$ (that is, symbols $3,\dots,n-1$ in \Cref{fig:lem-boundary-curve-1}). We claim that all these symbols must lie inside at most one segment. Indeed, assume there are $a,a' \in X_i$ in different segments. Since the subgraph induced by $X_i$ is connected, we have that $r(X_i)$ is a connected subset of the plane, and so $a,a'$ are connected by $r(X_i)$. However, $r(X_i)$ by definition does not intersect $r(i)$. 
    But this is impossible, since $r(i)$ is a Jordan curve connecting all appearances of symbol $i$, and so $r(i)$ splits all segments from each other, a contradiction. 
    (Note that by definition of $\trace(F)$, no curve can enter the inner of the boundary of $F$/the cycle $D$.)

    We can hence define $\overline s_i$ as the cyclic string starting and ending with symbol $i$ containing all segments except possibly the segment containing some symbol from $X_i$ (compare \Cref{fig:lem-boundary-curve-2}). Then all claimed properties of $\overline s_i$ are true and we have proven the claim.

    Next, we argue that for each pair $i \neq j$, it is impossible that $\overline s_i$ is completely contained in $\overline s_j$ (assuming $\overline s_i \neq \varepsilon$). By this, we mean that the corresponding cycle arcs are subsets of each other, i.e. $A_i \subseteq A_j$. Indeed, note that if this is the case, then $j \in \{i-1, i+1\}$. So for the sake of notation we can w.l.o.g.\ assume that $i = 1$ and $j = 2$. Then consider \Cref{fig:lem-boundary-curve-3}. Consider the symbols before and after $\overline s_2$.
    Since $\overline s_2$ starts and ends with a 2, and since $\overline s_1$ contains all occurrences of 1 and is contained in $\overline s_2$, both the symbol before and after $\overline s_2$ must be a 3. We use here that consecutive symbols on the cycle can differ by only one. 
    Consider now $\overline s_3$. The object $\overline s_3$ contains all occurrences of $3$, but cannot contain any 1 since $\overline s_3 \in \{2,3,4\}^\star$. So $\overline s_3$ must connect both 3's, but since $\overline s_1 \neq \varepsilon$ contains at least one 1, 
    the object $\overline s_3$ must wrap around the other way of the cycle. Then the whole cycle is contained in $\overline s_2 \cup \overline s_3$. But then $\trace(F) \in \{1,2,3,4\}^\star$. This contradicts our assumption that $\trace(F)$ contains at least 5 different symbols.

    We use the observation that no $\overline s_j$ contains any $\overline s_i$ in order to show the following structural insights:
    For all $i=1,\dots,n$, $\overline s_i = \varepsilon$ is actually impossible. To see this, let $\overline s_i = \varepsilon$ and w.l.o.g. $\overline s_{i+1} \neq \varepsilon$. Then symbol $i$ does not appear and $\overline s_{i+1} \in \{i+1, i+2\}$. Both the symbol before and after $\overline s_{i+1} \in \{i+1, i+2\}$ must be $i+2$, and so either $\overline s_{i+1}$ is contained in $\overline s_{i+2}$ or $\trace(F)$ has fewer than~$5$ distinct symbols, a contradiction. For the same reason, for each $i=1,\dots,n$ it is impossible that the symbols before and after $\overline s_i$ are both~$i-1$; or are both~$i+1$.
    
    Finally, we have arrived at the following conclusion: The strings $\overline s_1, \dots, \overline s_n$ must appear in exactly this order or its reflection along the cycle $D$. (By the order of appearance, we mean the order of the first time we encounter some element from $\overline s_i$ when going around the cycle). This follows since the two symbols before and after $\overline s_i$ are exactly $\{i-1,i+1\}$, and since $\overline s_i$ contains all symbols $i$.

    Finally, we show how to construct $s_1,\dots,s_n$ given $\overline s_1$. We simply define $s_i := \overline s_i \setminus \overline s_{i-1}$. Then since $\overline s_{i-1}$ contains all occurrences of $i-1$, we have that $s_i \in \{i,i+1\}$. As argued above, the strings $\overline s_1, \dots,\overline s_n$ are all nonempty, do not contain each other, and appear in exactly this order along the cycle. Then, we conclude that there is a cyclic shift and reflection $\tilde s$ of $s$ such that $s = s_1s_2\dots s_n$. The situation is depicted in \Cref{fig:lem-boundary-curve-4}. 
    
\end{proof}

\begin{figure}
     \centering
     \begin{subfigure}[b]{0.3\textwidth}
         \centering
         \includegraphics[scale=0.8,page=1]{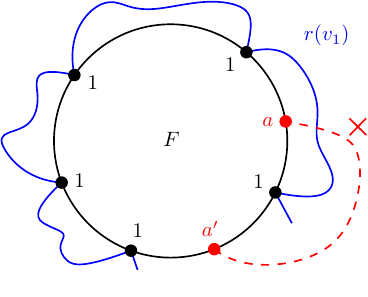}
         \caption{$a,a'$ cannot be in different segments.}
         \label{fig:lem-boundary-curve-1}
     \end{subfigure}
     \hfill
     \begin{subfigure}[b]{0.3\textwidth}
         \centering
         \includegraphics[scale=0.8,page=2]{figures/definition-s-bar.pdf}
         \caption{Definition of $\overline s_i$ for $i=1$.}
         \label{fig:lem-boundary-curve-2}
     \end{subfigure}
     \hfill
     \begin{subfigure}[b]{0.25\textwidth}
         \centering
         \includegraphics[scale=1]{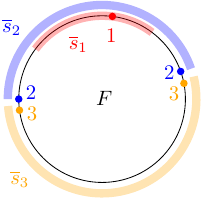}
         \caption{$A_1 \subseteq A_2$ is impossible}
         \label{fig:lem-boundary-curve-3}
     \end{subfigure}
     \hfill
     \vfill
    \begin{subfigure}[b]{0.3\textwidth}
         \centering
         \includegraphics[scale=1]{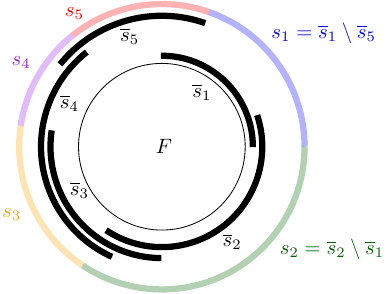}
         \caption{Definition of $s_i$ given $\overline s_1, \dots, \overline s_n$.}
         \label{fig:lem-boundary-curve-4}
     \end{subfigure}
        \caption{Arguments used in the proof of \Cref{lem:structureboundarycurve}}
        \label{fig:three graphs}
\end{figure}

The previous lemma will prove to be immensely helpful, since it has the following implication: If we have an induced cycle $C$, and we select a subset $Y$ of elements of pairwise distance at least two in $C$, then any cell in $r(C)$ that touches at least 5 elements must contain the elements of $Y$ in the same unique order along its trace (up to reflection and rotation).
Next, we argue that for our graph $G$ the realizations of certain vertices of the graph must be contained in some cell.

\subsection{Cell Lemma}

We start by introducing the notion of a \Connector.
Intuitively, a \Connector is just a set of vertices that separates
an induced cycle from the (connected) rest of the graph.

\begin{definition}\label{def:connectors}
    Let $G = (V,E)$ be a graph, and let $C\subseteq V$ be a set of vertices forming an induced cycle such that $G[V\setminus C]$ is connected.
    Let $D \subseteq V\setminus C$ be the neighborhood of $C$. If all of the following properties hold, we call $D$ the set of \textit{\Connector{s} of $C$}.
    \begin{itemize}[itemsep = 0pt]
        \item $D$ is an independent set.
        \item $G[V\setminus D]$ consists of the two connected components, $G[C]$ and $G[V \setminus (C\cup D)]$.
        \item Each $d\in D$ has exactly one neighbor in $C$, and for any two distinct $d,d'\in D$, these neighbors are distinct and non-adjacent.
    \end{itemize}
\end{definition}

This definition is illustrated in \Cref{fig:connectorsAndIntersectors}.

\begin{lemma}\label{lem:allinonecell}
    Let $C$ be an induced cycle of $G = (V,E)$ with \Connector{s} $D$, and let $r(V)$ be a geometric representation of $G$ by Jordan arcs. 
    Then $r(V\setminus (C\cup D))$ lies completely in one cell of the arrangement $r(C)$.
\end{lemma}
\begin{proof}
    This follows from the fact that $G[V\setminus (C\cup D)]$ is connected and is not adjacent to the induced cycle~$C$. Therefore, $r(V \setminus (C \cup D))$ is a connected subset of the plane that does not intersect $r(C)$, and so completely lies in one of its cells.
\end{proof}

Finally, note that if in the previous lemma $|D| \geq 5$, then the precondition of \Cref{lem:structureboundarycurve} is also satisfied.

\begin{figure}[tb]
    \centering
    
    \includegraphics{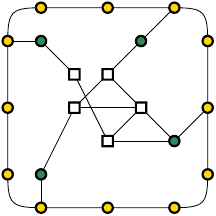} 
    \hspace{3cm}
    \includegraphics{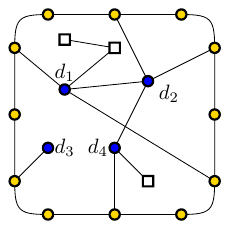}
    
    \caption{
    Left: 
    The yellow dots form an induced cycle. 
    The green vertices are connectors of this cycle and 
    the remaining white squares form the rest of the graph.
    Right: 
    The yellow dots form an induced cycle. 
    The blue dots are intersectors of this cycle. The cyclic order $o_G$ as used in \Cref{lem:sameorder} is $d_2,d_2,d_1,d_4,d_3,d_1$, when starting at the top left of the cycle and going clockwise. 
    }
    \label{fig:connectorsAndIntersectors}
\end{figure}

\subsection{Order Lemma}

In this section we show that
certain vertices neighboring an induced cycle
also need to intersect this cycle in the same order in any geometric representation.
We start with a definition of the setup.

\begin{definition}\label{def:intersectors}
    Let $G = (V,E)$ be a graph, and let $C\subseteq V$ be a set of vertices forming an induced cycle.
    Let $D \subseteq V\setminus C$ be a set of vertices with the following properties.
    \begin{itemize}[itemsep = 0pt]
        \item Each $d\in D$ has either one or two neighbors in $C$. 
        \item If $d\in D$ has two neighbors in $C$, then these are non-adjacent.
        \item For any two distinct $d,d'\in D$, their neighbors in $C$ are distinct and non-adjacent.
    \end{itemize}
    Then we call $D$ a set of \textit{intersectors of $C$}.
\end{definition}

This definition is illustrated in \Cref{fig:connectorsAndIntersectors}.
We now define two cyclic orders of intersections of intersectors with the cycle. The first order lives in the realm of the graph, while the second one is concerned with a concrete realization. The goal of this section is to prove that these orders are the same.

\begin{definition}[Graph Order of Intersectors]
    Let $G$ be a graph with $C$ an induced cycle on at least four vertices, and let~$D$ be a set of intersectors of~$C$. 
    When we travel along the cycle $C$ for one full rotation, 
    we write down the pairs $(c,d)\in C\times D$ such that $\{c,d\}\in E(G)$.
    This sequence of pairs defines the \textit{graph order of the intersectors} up to a cyclic shift and reflection.
\end{definition}

\begin{definition}[Geometric Order of Intersectors]
    Let $G$ be a graph with $C$ an induced cycle on at least four vertices, and let~$D$ be a set of intersectors of~$C$. 
    Let $r$ be a realization of $G$ by unit segments, and let $b$ be the interior boundary curve of $r(C)$. More generally, if $r(C)$ is a generalization with Jordan curves, let $b$ be the boundary curve of any cell $F$ such that $\trace(F)$ has at least 5 distinct symbols.
    When we travel along  $b$ for one full rotation, 
    we write down the pairs $(c,d)\in C\times D$ such that $r(d)$ intersects $b$ where $b$ is tracing the boundary of the cell $F$ in $r(c)$. From consecutive copies of the same pair we only keep one.
    This sequence of pairs defines the \textit{geometric order of the intersectors} up to a cyclic shift and reflection.
\end{definition}

    The graph order is illustrated in \Cref{fig:connectorsAndIntersectors}, while the geometric order is illustrated in \Cref{fig:orderalongcurve}.

\begin{lemma}\label{lem:sameorder}
    Let $r$ be a realization of a graph $G$ with an induced cycle $C$ and intersectors $D$. If for every intersection of some $r(d)$ for $d\in D$ with some $r(c)$ for $c\in C$ we have that $r(d)$ also intersects the boundary curve $b$ close to that intersection, then the geometric order and the graph order of the intersectors coincide.
\end{lemma}

\begin{figure}[tb]
    \centering
    \includegraphics{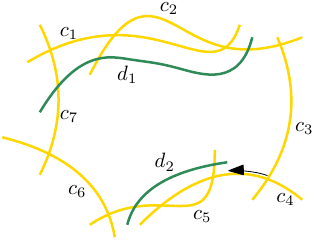}
    \caption{A realization $r$ of an induced cycle $C$ (yellow Jordan arcs) and intersectors $\{d_1,d_2\}$ of $C$ (green Jordan arcs). The cyclic order $o_r$ along the interior boundary curve as used in \Cref{lem:sameorder} is $(c_5,d_2),(c_7,d_1),(c_2,d_1)$, when starting at the bottom right and going clockwise (see black arrow).}
    \label{fig:orderalongcurve}
\end{figure}

\begin{proof}
    We first claim that every pair $(c,d)$ occurs in both orders exactly once if $\{c,d\}\in E(G)$, and zero times otherwise: In the graph order this holds by definition. Furthermore, the assumption of this lemma guarantees that any pair $(c,d)$ in the graph order of intersectors also occurs in the geometric order of intersectors. We thus need to show only that in the geometric order no pair occurs multiple times. A pair could occur more than once only if there is another pair $(c',d')$ between these occurrences. However, since $c$ is crossed only by $d$, and no neighbor of $c$ in $C$ is crossed by any intersector, \Cref{lem:structureboundarycurve} guarantees that no such pair $(c',d')$ can occur between the two occurrences of $(c,d)$.

    Next, we show that the pairs are also ordered the same way.
    Since for any two pairs $(c,d)$ and $(c',d')$ in the order the two vertices $c,c'$ are distinct and non-adjacent, \Cref{lem:structureboundarycurve} guarantees that the geometric order of the intersectors must respect the ordering of the $c_i$ along the cycle. Similarly, the graph order of intersectors must respect this ordering by definition. Thus, the two orders are the same. 
\end{proof}    

Note that the definition of the geometric order and \Cref{lem:sameorder} also work for the exterior boundary curve instead of the interior one.

\section{\ER-Membership}\label{sec:membership}
In this section, we show the \ER-membership parts of \Cref{thm:unit,thm:polylines}. 

There are two standard ways to establish \ER-membership.
The naive way is to encode the problem at hand as an ETR-formula. 
The second method describes a real witness and a real verification algorithm, similar as to how one can prove \NP-membership using a combinatorial verification algorithm. 
We describe both approaches.

In order to present the naive approach, we restrict our attention to only unit segments. However a similar technique also easily works for polylines. Let $G = (V,E)$ be a graph.
For each vertex $v\in V$, we use four variables $\var{v1},\var{v2},\var{v3},\var{v4}$ 
to describe the coordinates of the endpoints of a unit segment realizing $v$. 
We can construct ETR-formulas $\unit$ and $\intersection$ that test whether a segment has unit length and whether two segments  intersect, respectively.
The formula $\varphi$ for \UnitRec consists of the three parts:
\[\bigwedge_{v\in V} \ \unit (\var{v1},\var{v2},\var{v3},\var{v4}), \]
\[\bigwedge_{uv\in E} \  
\intersection
(\var{u1},\var{u2},\var{u3},\var{u4},\var{v1},\var{v2},\var{v3},\var{v4}),\]
and 
\[\bigwedge_{uv\not \in E} \ \lnot \, \intersection(\var{u1},\var{u2},\var{u3},\var{u4},\var{v1},\var{v2},\var{v3},\var{v4}).\]
The $\unit$ formula can be constructed using the formula for the Euclidean norm.
To construct the \intersection formula it is possible to use the orientation test: The orientation test formula checks whether a given  ordered triple of points is oriented clockwise, or counter-clockwise. Using multiple orientation tests on the endpoints of two segments, one can determine whether the segments intersect.
The orientation test itself can be constructed using a standard determinant test.
This finishes the description of the formula $\varphi$ and establishes \ER-membership of \UnitRec. 

Although all of these formulas are straightforward to describe, things get a bit lengthy (especially in the case of polylines) and we do hide some details about the precise polynomials. 
We therefore also describe the second approach using real witnesses and verification algorithms.
To use this approach we need to first introduce a different characterization of the complexity class \ER. 
Namely, an algorithmic problem is in \ER if and only if we can provide a \emph{real verification algorithm}~\cite{SmoothingGap}.
A real verification algorithm $A$ for a problem $P$ takes as input an instance $I$ of $P$ and a polynomial-size real-valued witness $w$.
$A$ must have the following properties: In case that $I$ is a yes-instance, there exists a $w$ such that $A(I,w)$ returns yes. In case that $I$ is a no-instance, $A(I,w)$ will return no, for all possible $w$.
Note that this is reminiscent of the definition of \NP using a verifier algorithm. 
There are two subtle differences: 
The first one is that $w$ is allowed to contain \emph{real} numbers and discrete values, not only bits.
The second difference is that $A$ runs on the \emph{real RAM} instead of a Turing machine. This is required since a Turing machine is not capable of dealing with real numbers.
It is important to note that $I$ itself does not contain any real numbers.
We refer to Erickson, van der Hoog, and Miltzow~\cite{SmoothingGap} for a detailed definition of the real RAM.

Given this alternative characterization of \ER, it is now very easy to establish \ER-membership of \UnitRec and \PolyRec:
We merely need to describe the witness and the verification algorithm.
The witness is a description of the coordinates of the unit segments (or polylines, respectively) realizing the given graph.
The verification algorithm checks that the realizations of vertices do or do not intersect one another, and in the case of unit segments, also checks that all segments have the correct length.
Note that we sweep many details of the algorithm under the carpet. 
However, algorithms are much more versatile than formulas and it is  a well-established fact that algorithms are capable of all types of elementary operations needed to perform this verification. 

\section{\ER-Hardness}

In this section we prove \ER-hardness, first for unit segments, then for \polylines.
The reduction for \polylines builds upon the reduction 
for unit segments, and we will only highlight the differences.
\subsection{\ER-Hardness for Unit Segments}

We show \ER-hardness of \UnitRec by a reduction from \stretchability.

\begin{figure}[tb]
    \centering
    \includegraphics{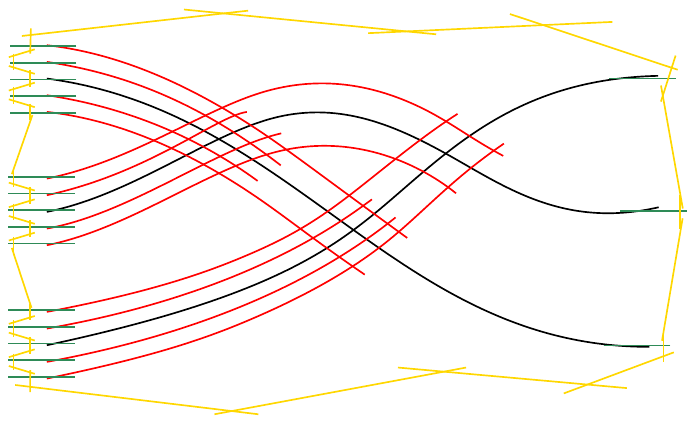}
    \caption{A pseudoline arrangement (black), enhanced with probes (red), connectors (green) and a cycle (yellow).
    }
    \label{fig:arrangementenhanced}
\end{figure}


We present the reduction in three steps. 
First, we show how we construct a graph $G$ from a pseudoline arrangement \A. 
Next we show completeness, i.e., we show that if \A is stretchable then $G$ can be represented using unit segments.
At last, we show soundness, i.e., we show that if $G$ can be represented using unit segments then \A is stretchable.
For this part we will use the lemmas from \Cref{sec:Cycle}.

\paragraph{Construction.}
Given a pseudoline arrangement \A of $n$ pseudolines, we construct a graph $G$ 
by enhancing the arrangement \A with additional Jordan arcs. 
Then we define $G$ to be the intersection graph of the pseudolines (which are also Jordan arcs) and all our added arcs.
See \Cref{fig:arrangementenhanced} for an illustration.

First, we add so called \emph{probes} to our pseudolines. 
A probe of \emph{depth} $k$ of pseudoline $\ell$ is a Jordan arc that starts on 
the left vertical line, and follows $\ell$ closely through the arrangement
until it has intersected $k$ other pseudolines (and all probes of other pseudolines that reach this intersection). 
For each pseudoline $\ell$, we add $2(n-1)$ probes: Above and below $\ell$, we add one probe each for each depth $1\leq k \leq n-1$.
This gives a total of $p := 2n(n-1)$ probes.
The probes are sorted according to their depth, with the probes of smallest depth situated closest to $\ell$. 
Note that so far we have added $n + p$ Jordan arcs.
We now create \emph{connectors} for each probe and pseudoline. A connector is a short Jordan arc added to the left and/or right end of another arc. For probes, we only add connectors at the left end, while pseudolines get connectors at both ends.
Note that we thus add $d := 2n+p$ connectors.
Finally, we add Jordan arcs forming a cycle of length $2d + 6$ around the current arrangement. On the left and right side, this cycle is placed in such a way that every second arc of the cycle intersects a connector, in the correct order. The eight remaining arcs of the cycle connect the left and right side of the cycle, using four arcs on the top and four arcs on the bottom.

We call the collection of 
all these arcs (including the pseudolines) the \textit{enhanced arrangement}.
The graph $G$ is the intersection graph of this arrangement.
Note that $G$ could also be described purely combinatorially, and it can be constructed in polynomial time.

\begin{figure}[tb]
    \centering
    \includegraphics{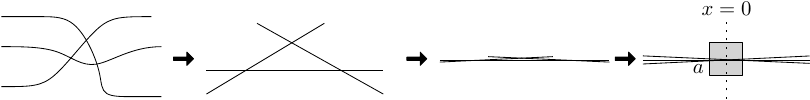}
    \caption{Left to Right: pseudoline arrangement \A; stretched line arrangement \L, squeezed line arrangement to have small slopes; squeezed to have all intersections in a box of side length $a$.}
    \label{fig:squeeze}
\end{figure}

\paragraph{Completeness.}
In this paragraph, we show 
that if  \A is stretchable then $G$ is realizable by unit segments.
We assume that \A is stretchable and show how to place each unit segment realizing $G$.
We call the segments representing the pseudolines, probes, connectors, and the cycle
\textit{important segments}, \textit{probe segments}, \textit{connector segments}, and \textit{cycle segments}, respectively.

As the pseudoline arrangement \A is stretchable, there exists a combinatorially equivalent line arrangement \L. 
The arrangement~\L can be compressed (scaled down along the vertical axis) such that the slopes of all lines lie in some small interval $[-a,a]$ for, say, $a = 1/20$. 
Additionally, we can move and scale \L even more to ensure that all 
intersection points lie in the square $[-a,a]^2$.
See \Cref{fig:squeeze} for an illustration.

We now truncate all the lines of \L to get unit segments that have the same intersection pattern as the corresponding lines. 
The truncation is performed symmetrically around the vertical line $x=0$.
The resulting unit segments are our \emph{important segments}.

For the next step we consider an important segment $s$ and construct its corresponding \emph{probe segments}. 
We place all probe segments parallel to its important segment $s$, with sufficiently small distance to~$s$ and to each other.
The probe segments are placed as far towards the left as possible, while still reaching the intersections of important segments they need to reach. Since by construction of the scaled line arrangement \L, all intersections lie within $[-a,a]^2$ and the probes can go only until there, each probe segment and its corresponding important segment are almost collinear but shifted by roughly $0.5$ longitudinally.
See \Cref{fig:Unit-Completeness} for an illustration of the placement of the probe segments.

Next we need to describe the placement of the \emph{connector segments}. 
Note that on the right side, we only have important segments. 
We add all the connector segments in such a way that they
lie on the same line as the important/probe segment they attach to.
The connectors can overlap with the segments they attach to for a large part since the first intersection of any probe and important segments only occurs 
in the square $[-a,a]^2$. 
This allows us to place our connectors such that all connectors on the left (right) side of the drawing end at the same left (right) $x$-coordinate. 

\begin{figure}[bt]
    \centering
    \includegraphics{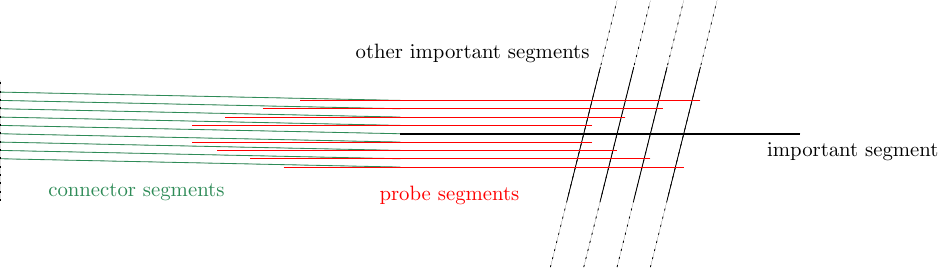}
    \caption{Given the line arrangement \L, we can construct the probe segments and connector segments using unit segments.
    The connector segments are illustrated slightly tilted, for the purpose of readability.
    They should be parallel to the segment that they connect to.
    }
    \label{fig:Unit-Completeness}
\end{figure}

\begin{figure}[bt]
    \centering
    \includegraphics[page = 2]{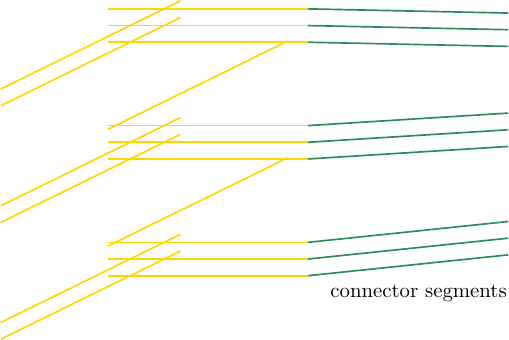}
    \caption{We can place all the cycle segments correctly on the left and on the right using a simple sawtooth pattern.
    }
    \label{fig:Connector-Segments}
\end{figure}

Finally, we can draw the cycle segments. For this, we simply make a sawtooth pattern on the left and on the right.
See \Cref{fig:Connector-Segments} for an illustration.
In  our sawtooth pattern, every second cycle segment is horizontal,
and all other cycle segments have the same, slightly positive slope. The horizontal segments attach to the connector segments.
As the important segments all have a very small slope, we can never run into a situation where two horizontal cycle segments would be too far away from each other to be connected.
We connect the left and right sawtooth patterns to close the cycle, using our eight additional cycle segments. 
(For the following reason, eight segments are sufficient for this task: Since the left and right sawtooth patterns are connected by a path that consists of a left connector segment, an important segment, and a right connector segment, their distance is at most $3$, allowing us to form both the top and the bottom boundary of the cycle with at most four cycle segments each.)

\paragraph{Soundness.}
In this paragraph, we show that
 $G$ being realizable using unit segments implies that \A is stretchable.
We thus assume that there exists a realization $r$ of $G$ by unit segments. 
Similarly to the last paragraph, we denote the sets of vertices representing the pseudolines, probes, connectors, and the cycle by \emph{important vertices $I$}, \emph{probe vertices $P$}, \emph{connector vertices $D$}, and \emph{cycle vertices $C$}.

Note that $C$ forms an induced cycle in $G$ and $D$ is a set of connectors of $C$ as in \Cref{def:connectors} (thus motivating the name). 
By \Cref{lem:allinonecell}, $r(C)$ splits the plane into two cells, and $r(I\cup P)$ is contained completely in one of these two cells. Without loss of generality, we assume $r(I\cup P)$ is contained in the inner (bounded) cell, however all following arguments would also work with regards to the outer cell. 
Thus, every segment in $r(D)$ 
intersects the cycle~$r(C)$ from the inside, i.e., it
intersects the interior boundary curve.
We can thus apply \Cref{lem:sameorder}, and get that $r(D)$ is ordered along the interior boundary curve of~$C$ in the same way (up to cyclic shift and reflection) as
it is in the enhanced arrangement from which we defined $G$.
Specifically, we know that the important segments and probe segments are ordered  as in our enhanced arrangement; see \Cref{fig:arrangementenhanced}.

We now claim that the arrangement of the important segments $r(I)$ is combinatorially equivalent to \A.
In order to show this claim, our proof strategy is to define two orientations: One orientation $O_{\A}$ of the pseudolines in the pseudoline arrangement $\A$, 
and one orientation $O_r$ of the important segments $r(I)$ in the realization $r$.
We show that the order in which pseudolines intersect other pseudolines with respect to $O_\A$ is equal to the order in which important segments are intersected by other important segments with respect to $O_r$. 
If we can show this claim, we are done, since the order of important segments and unit segments along the interior boundary in $r(I)$ is the same as in our enhanced arrangement. Therefore, (by extending the important unit segments $I$ to infinite lines) we obtain that $\A$ is stretchable.

It remains to show the claim. The orientations $O_\A$ and $O_r$ are defined as follows: In the orientation $O_\A$, 
every pseudoline is simply oriented from left to right (where left and right correspond to the directions in \Cref{fig:arrangementenhanced}). The orientation $O_r$ is defined as follows: 
Consider some important vertex $v \in I$, and its corresponding important segment $r(v)$. 
Note that $v$ has two neighbors $d_1,d_2$ in $G$ such that $d_1,d_2 \in D$ are connector vertices.
We can assume that $d_1$ is a connector segment corresponding to a connector on the left (of \Cref{fig:arrangementenhanced}) and $d_2$ is a connector segment corresponding to a connector on the right (of \Cref{fig:arrangementenhanced}).
Since $v,d_1$ are neighbors in $G$, the unit segment $r(v)$ has a crossing $c_1$ with $r(d_1)$. Likewise it has a crossing $c_2$ with $r(d_1)$. In the orientation $O_r$ we orient the segment $r(v)$ such that $c_1$ comes before $c_2$. (Note that if the unit segments $r(d_i), r(v)$ are not parallel for $i=1,2$ then the crossings $c_1,c_2$ are unique and $c_1 \neq c_2$, and so $O_r$ is well-defined. In the case that $r(d_i), r(v)$ are parallel it is still well-defined since $r(d_1),r(d_2)$ do not intersect.)

Now we are ready to show the claim. We pick three pseudolines $\ell$, $a$, and $b$ 
such that $\ell$ intersects $a$ before~$b$ in the orientation $O_\A$. 
We prove that in the orientation $O_r$, the unit segment $r(\ell)$ must also intersect the segment $r(a)$ before $r(b)$. (In order to simplify the notation, we use the same symbol for a pseudoline, as well as its corresponding vertex in $G$).

Now, suppose for the purpose of a contradiction that $r(\ell)$ intersects $r(b)$ before $r(a)$. 
Consider the following curves:
\begin{itemize}[itemsep = 0pt]
    \item the outermost probe segments of $\ell$,
    \item their connector segments,
    \item a part of the interior boundary curve of the cycle segments, and
    \item the important segment $b$.
\end{itemize}
Since these segments form an induced cycle in the graph $G$, these segments bound a cell $E$. In this cycle, $\ell$ and its connector segment $d_1$ (the one attaching to the cycle between the probes) form a path of length two between two vertices of the cycle. Since $d_1$ is a connector segment, due to the previous observation about the fixed order of connector segments, $d_1$ lies between the outer two connector segments. Thus $r(d_1) \cup r(\ell)$ split $E$ into two parts $E_1,E_2$. See \Cref{fig:cells} for an illustration. Note that $r(\ell)$ must be oriented from the end contained in $E$ towards the other end since the intersection of $r(d_1)$ and $r(\ell)$ lies in $E$, but the intersection of $r(d_2)$ and $r(\ell)$ cannot. Since we assume that $r(\ell)$ intersects $r(b)$ before $r(a)$, the intersection of $r(a)$ and $r(\ell)$ thus lies outside of $E$.

We consider the two probe segments $p$ and $q$ of $\ell$, which correspond to the intersection with $a$. 
These segments are attached to the cycle between the outermost probe segments of $\ell$, and $\ell$ itself. Furthermore, $p$ and $q$ are not intersecting any segment bounding $E_1$ or $E_2$, in particular not $r(b)$.
Thus the probe segments $p$ and $q$ are both completely contained in the interiors of $E_1$ and $E_2$, respectively. 
However, $r(a)$ can intersect the interior of only one of $E_1$ and $E_2$, but not both, since $r(a)$ and $r(\ell)$ are line segments and their single intersection is assumed to lie outside of $E$.
Since both $p$ and $q$ must intersect $r(a)$, we arrive at the desired contradiction. We thus conclude that \A is stretchable, finishing the proof of \Cref{thm:unit}.

\begin{figure}[tb]
    \centering
    \includegraphics[page =2]{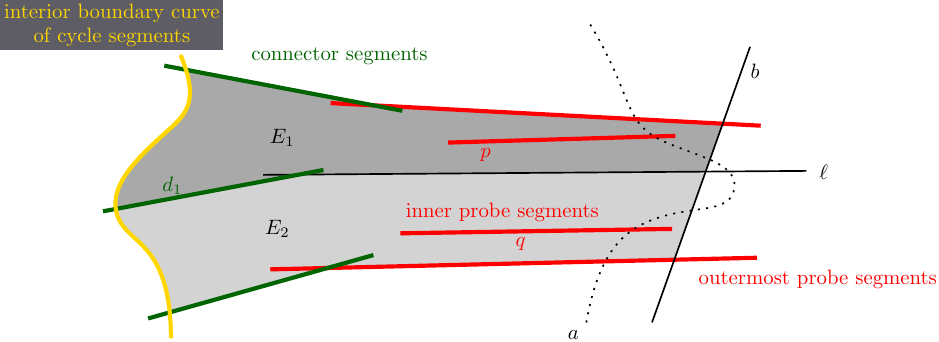}
    \caption{
    Suppose  for the purpose of contradiction that $\ell$ intersects $a$ after $b$.
    Then the important segment $a$ can only intersect one of the cells $E_1$ and $E_2$. However, $E_1$ and $E_2$ contain the inner probe segments $p$ and $q$ in their respective interiors. Since $a$ must intersect $p$ and $q$, we have a contradiction.}
    \label{fig:cells}
\end{figure}

\subsection{\ER-Hardness for \Polylines}

We study \polylines with $k$ bends and we assume $k>0$ is a fixed constant (note that the case $k=0$ corresponds to segments, for which the recognition problem has already been proven \ER-complete in \cite{kratochvil1994intersection}). 
Specifically, we will show that \PolyRec is \ER-complete.

The proof for \polylines works very similarly to the proof for unit segments. 
Since the family of \polylines is a strict superset of the family of unit segments, 
most of our additional work is on the soundness of the proof. 
To be able to ensure soundness, our construction 
of the graph $G$ will make sure that the polylines realizing the pseudolines cannot contain any bends in some region encoding the pseudoline arrangement \A in any realization of $G$.
With this property, the argument for soundness (as in the proof of \Cref{thm:unit}) will carry over straightforwardly.

\paragraph{Construction.}
As in the construction for unit segments, we enhance our arrangement \A of $n$ pseudolines using additional Jordan arcs, and let our graph $G$ be the intersection graph of the enhanced arrangement.

The whole construction is illustrated in \Cref{fig:polylinearrangementenhanced}. To create the enhanced arrangement, we first create the \emph{frame}:
We create $2k+2$ vertical chains of $4n+1$ line segments each. We call the set of segments forming the $i$-th such vertical chain $C_i$, for $i = 1\dots, 2k+2$, numbered from left to right. We then connect these vertical chains by two horizontal chains, one at the top and one at the bottom. More precisely, the top chain connects all the top segments of $C_i$ for $i = 1\dots, 2k+2$ such that two segments have distance three. The bottom chain behaves analogously.
All of the segments forming these chains are called \emph{frame segments}, and they together bound $2k+1$ bounded cells. We call the leftmost such cell the \emph{canvas}.

We now place our pseudoline arrangement in the canvas. 
For every pseudoline $p$, we first introduce a parallel \emph{twin} pseudoline $p'$ closely below $p$. Since \A is assumed to be simple, every pseudoline has the same intersection pattern with all the other pseudolines as its twin. 

Next, we introduce probes from the left of our pseudoline arrangement, as we did in the proof of \Cref{thm:unit}. Note that every pair of twin pseudolines shares one set of probes. We attach the probes as well as the pseudolines and their twins to $C_1$ using connector arcs, making sure that each connector intersects a unique arc of $C_1$, and that none of these arcs are intersecting.

We now want to weave the pseudolines and their twins through $C_2,\ldots, C_{2k+1}$. To do this, we split each set $C_i$ into $n$ \emph{lanes} $L_{i,1},\ldots,L_{i,n}$: A lane is a set of three consecutive arcs of $C_i$. 
The lanes $L_{i,j}$ are pairwise disjoint, and no arc of any lane $L_{i,j}$ may intersect an arc of a lane $L_{i,j'}$. 
Note that the lanes are ordered along $C_i$ with $L_{i,1}$ being the topmost lane and $L_{i,n}$ the bottommost. We also number our pseudolines from $p_1$ to $p_n$ by their order at the right end of the arrangement, from top to bottom. For every pseudoline $p_j$, we now extend $p_j$ and its twin $p_j'$ as follows: For each even $i$, $p_j$ goes through the top arc of the lane $L_{i,j}$, and for each odd $i$, it goes through the bottom arc of $L_{i,j}$. The twin $p_j'$ is extended in the opposite way, going through the top arc of $L_{i,j}$ if $i$ is odd, and through the bottom arc for $i$ even. 

At the very end, at $C_{2k+2}$, we use a connector arc to attach each pseudoline $p_j$ to the top element of $L_{2k+2,j}$, and its twin $p_j'$ to the bottom element. 

\begin{figure}[tb]
    \centering
    \includegraphics{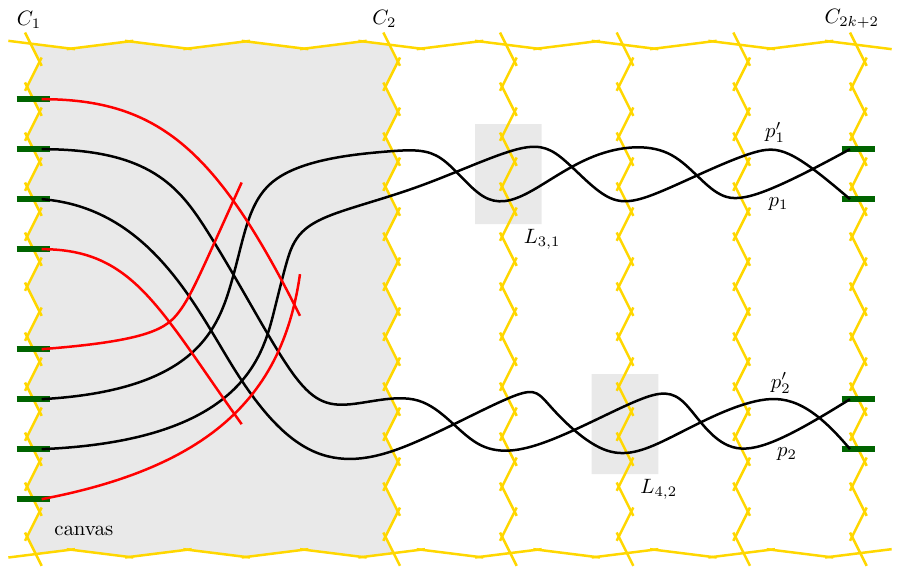}
    \caption{A doubled pseudoline arrangement (black) enhanced for $k=2$ with probes (red), connectors (green) and the frame (yellow). The shaded region is the canvas.
    }
    \label{fig:polylinearrangementenhanced}
\end{figure}

\paragraph{Completeness.} 
Given a line arrangement \L combinatorially equivalent to \A, we want to show that $G$ is realizable by \polylines. 
It is easy to see that given \L, the canvas and its contents (probes, connectors, pseudolines, and pseudoline twins) can be realized even with line segments (i.e., using $0$ bends), as we argued already in the proof of \Cref{thm:unit}. The only remaining difficulty is to argue that the extended pseudolines  weaving through $C_2,\ldots,C_{2k+1}$ and finally attaching to $C_{2k+2}$ can be realized using at most $k$ bends per pseudoline and twin.

Since between every two lanes $L_{i,j}$ and $L_{i,j+1}$ there is at least one \polyline that is not part of any lane, we can make the distance between these lanes arbitrarily large. We thus need to show only that for a single pair $p_j$, $p_j'$, we can weave the \polylines through $L_{2,j},\ldots,L_{2k+2,j}$. The construction is illustrated in \Cref{fig:weaving}: The two \polylines exit the canvas in parallel.
Then, going from left to right, we first use all bends of $p_j$, leaving its twin straight. After having gone through $L_{2,j},\ldots,L_{k+2,j}$, we leave $p_j$ straight and use the bends of the twin to go through the remaining lanes $L_{k+3,j},\ldots,L_{2k+1,j}$ and to have the right ordering to attach to $C_{2k+2}$ using a connector.

\begin{figure}[tb]
    \centering
    \includegraphics{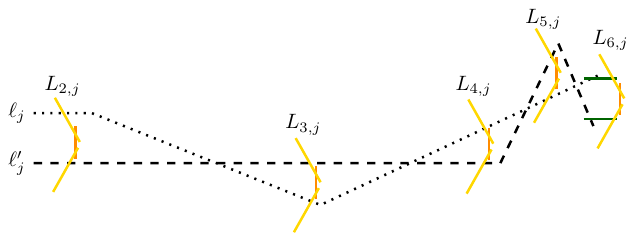}
    \caption{Weaving the two copies (dashed and dotted) of a pseudoline through the $2k$ chords with at most $k$ bends per \polyline. Illustrated for $k=2$.}
    \label{fig:weaving}
\end{figure}

\paragraph{Soundness.}
Given a realization of the graph $G$ with \polylines, we want to argue that \A is stretchable. To achieve this, we first prove some structural lemmas
about all realizations of $G$.

\begin{lemma}
\label{lem:polylines-contained-in-canvas}
    In every realization $r(G)$ of $G$, the polylines realizing all probes, pseudolines, and internal vertices of $C_2,\ldots,C_{2k+1}$ must lie in the same cell of the arrangement created by the Jordan curves corresponding to the cycle bounded by $C_1$, $C_{2k+2}$, the top chain, and the bottom chain.
\end{lemma}
\begin{proof}
    The induced subgraph of this set of vertices is connected. The connectors attaching them to the outermost cycle together with  the outermost segments of $C_2,\ldots,C_{2k+1}$ thus form a set of connectors as defined in \Cref{def:connectors}. Thus the lemma follows from \Cref{lem:allinonecell} immediately.
\end{proof}

Let us call the cell defined by \Cref{lem:polylines-contained-in-canvas} the cell $F_\text{region}$. We again assume without loss of generality that $F_\text{region}$ is not the outer cell. This is because if it is the outer cell, we can perform a circular inversion. Under such a circular inversion all Jordan curves remain Jordan curves and the transformed drawing is still a valid geometric representation.
We now observe that $\trace(F_\text{region})$ has at least 5 different elements, and so \Cref{lem:structureboundarycurve} is applicable, and we know that the order of elements appearing along the boundary of $F_\text{region}$ is almost fixed. 
In particular, this implies that the chains $C_2,\dots,C_{2k+1}$ appear in order from left to right in the cell $F_\text{region}$, subdividing it into a total of $2k+1$ disjoint sub-regions, where the first such sub-region is the canvas. 
(Note that locally, some small cells with only few elements on the boundary might be formed by the Jordan curves, but these do not influence the correctness of our argument.)  
We furthermore conclude that the curves $r(p_j)$ and $r(p_j')$ must traverse $F_\text{region}$ and while doing so cross every of the $2k+1$ sub-regions.
We can thus apply \Cref{lem:sameorder} to the sub-regions to enforce the ordering of crossings along them.
From this, similarly as in the proof of soundness for unit segments (note that that proof did not use the fact that the line segments were unit), we wish to get that the arrangement obtained by picking either $r(p_j)$ or $r(p_j')$ for every $j\in [n]$ restricted to the canvas has the same combinatorial structure as \A. However, for this it remains to show that at least one of $r(p_j)$ and $r(p_j')$ must be a straight line within the canvas. To show this, we first show that the two polylines $r(p_j)$ and $r(p_j')$ must cross often:

\begin{lemma}\label{lem:crossingpercell}
    In every sub-region defined by the cycle enclosed by two consecutive sets $C_i,C_{i+1}$ for $2\leq i\leq 2k+1$, we have that for all $j\in\{1,\ldots,n\}$ the curves $r(p_j)$ and $r(p_j')$ must cross.
\end{lemma}
\begin{proof}
    By \Cref{lem:sameorder} and by construction of $G$, the geometric order of $r(p_j)$ and $r(p_j')$ along the interior boundary curve of the cell bounded by $C_i$ and $C_{i+1}$ must be alternating (non-nesting). Thus, the two \polylines $r(p_j)$ and $r(p_j')$ must cross within this cell. 
\end{proof}

\begin{lemma}
    \label{lem:consecutive-intersections}
    Let $c$ and $d$ be two \polylines that intersect in exactly $t$ points $i_1,\ldots,i_t$, with both polylines visiting the intersection points in this order.
    Then $c$ and $d$ have at least $t-1$ bends in total between the first and the last intersection.
\end{lemma}
\begin{proof}
    It is easy to see that there must be at least one bend between any two consecutive intersection points.
    See \Cref{fig:consecutive-intersections} for an illustration.
\end{proof}

\begin{figure}
    \centering
    \includegraphics{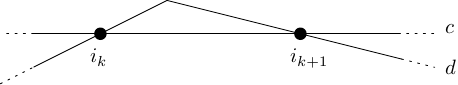}
    \caption{It is easy to see that the two \polylines must have at least one bend between the two intersection points.}
    \label{fig:consecutive-intersections}
\end{figure}

We now finally get our desired lemma:
\begin{lemma}
    \label{lem:straight}
    At least one of $r(p_j)$ and $r(p_j')$ is a straight line within the canvas.
\end{lemma}
\begin{proof} 
    The two twin polylines 
    must have at least $2k$ intersection points that occur in the same order along the two twins: to find these points, we pick one per cell enclosed between $C_i,C_{i+1}$ for $2\leq i\leq 2k+1$ as guaranteed by \Cref{lem:crossingpercell}.
    Thus by \Cref{lem:consecutive-intersections}, 
    the two \polylines have at least $2k-1$ bends in total outside of the canvas.
    This implies that at least one of them is straight inside the canvas.
\end{proof}

We conclude that \A must be stretchable, finishing the proof of \Cref{thm:polylines}.

\newpage

\bibliographystyle{plainurl}
\bibliography{ETR,ref}

\end{document}